\documentclass[11pt]{article}

\usepackage{amsmath}
\usepackage{amsthm}
\usepackage{aliascnt}
\usepackage{authblk}

\usepackage{amssymb}
\usepackage{graphicx}
\usepackage{multicol}
\usepackage{multirow}
\usepackage{color}
\usepackage{bm}
\usepackage{bbm}
\usepackage[letterpaper,margin=1in,bottom=1in]{geometry}
\usepackage[utf8]{inputenc}
\usepackage[english]{babel}
\usepackage{mathtools}
\usepackage{aliascnt}

\newcommand{\ketbra}[2]{\ensuremath{\ket{#1}\!\bra{#2}}}

\DeclarePairedDelimiter\ket{\lvert}{\rangle}
\DeclarePairedDelimiter\bra{\langle}{\rvert}

\newcommand{\proj}[1]{\ketbra{#1}{#1}}
\newcommand{\E}{\mathbb{E}}
\newcommand{\Prb}{\operatorname{Pr}}

\newcommand{\mathds}[1]{\mathbbm{#1}}
\newcommand{\mathup}[1]{\mathrm{#1}}

\usepackage{algorithm}
\usepackage{algpseudocode}

\usepackage{tabularx}
\usepackage{booktabs}
\usepackage{enumitem}
\usepackage{threeparttable}
\usepackage{adjustbox}
\usepackage{xspace}

\usepackage{tikz}
\usetikzlibrary{arrows}
\usetikzlibrary{quantikz2}
\usepackage[
  pagebackref,
  colorlinks,
  linkcolor=magenta,
  citecolor=magenta
]{hyperref}
\usepackage[capitalize,noabbrev]{cleveref}

\newtheorem{theorem}{Theorem}[section]
\newaliascnt{axiom}{theorem}

\aliascntresetthe{axiom}
\newaliascnt{lemma}{theorem}
\newtheorem{lemma}[lemma]{Lemma}
\aliascntresetthe{lemma}
\newaliascnt{corollary}{theorem}
\newtheorem{corollary}[corollary]{Corollary}
\aliascntresetthe{corollary}
\newaliascnt{proposition}{theorem}
\newtheorem{proposition}[proposition]{Proposition}
\aliascntresetthe{proposition}
\newaliascnt{fact}{theorem}

\aliascntresetthe{fact}

\newtheorem{definition}{Definition}[section]

\newtheorem{remark}{Remark}[section]

\crefname{theorem}{Theorem}{Theorems}
\Crefname{theorem}{Theorem}{Theorems}
\crefname{axiom}{Axiom}{Axioms}
\Crefname{axiom}{Axiom}{Axioms}
\crefname{lemma}{Lemma}{Lemmas}
\Crefname{lemma}{Lemma}{Lemmas}
\crefname{corollary}{Corollary}{Corollaries}
\Crefname{corollary}{Corollary}{Corollaries}
\crefname{proposition}{Proposition}{Propositions}
\Crefname{proposition}{Proposition}{Propositions}
\crefname{fact}{Fact}{Facts}
\Crefname{fact}{Fact}{Facts}
\crefname{remark}{Remark}{Remarks}
\Crefname{remark}{Remark}{Remarks}

\newcommand{\Qreg}{\mathup{Q}}
\newcommand{\Wreg}{\mathup{W}}
\newcommand{\Hdbreg}{\mathup{H}}
\newcommand{\Gdbreg}{\mathup{G}}

\newcommand{\calA}{\mathcal{A}}
\newcommand{\calB}{\mathcal{B}}

\newcommand{\FilteredSum}{\emph{$\ell$-Sum-Filtered Parity}\xspace}

\title{Quantum Query Advantage Requires Space}
\author[1,2]{Minbo Gao\thanks{%
  \href{mailto:gaomb@ios.ac.cn}{\nolinkurl{gaomb@ios.ac.cn}} or
  \href{mailto:gmb17@tsinghua.org.cn}{%
    \nolinkurl{gmb17@tsinghua.org.cn}}.}}
\author[3]{Zhengfeng Ji\thanks{%
  \href{mailto:jizhengfeng@tsinghua.edu.cn}{%
    \nolinkurl{jizhengfeng@tsinghua.edu.cn}}.}}
\author[3]{Ziyi Xie\thanks{%
  \href{mailto:xie-zy21@tsinghua.edu.cn}{%
    \nolinkurl{xie-zy21@tsinghua.edu.cn}}.}}
\affil[1]{Institute of Software, Chinese Academy of Sciences.}
\affil[2]{University of Chinese Academy of Sciences.}
\affil[3]{Tsinghua University.}
\date{}

\begin{document}

\maketitle

\begin{abstract}
Hao, Huang, and Liu~\cite{HHL26} recently showed that optimal quantum query complexity may require large workspace even for short-output problems, and asked whether a quantum query advantage over classical computation can itself require space. We resolve this question by exhibiting an explicit total Boolean function with quantum query complexity
$Q=\widetilde{\Theta}(M)$,
and randomized query complexity
$R=\Theta(M^{21/20})$,
whereas, for every fixed $0<\eta<1/100$ and $ S\le O(M^{1/100-\eta})$,
its $S$-space quantum query complexity satisfies
 $Q_S=\omega(M^{21/20})$.
Consequently,
\[
Q<R<Q_S,
\]
so the unrestricted quantum query advantage disappears under sufficiently small workspace.

The separation is obtained through a one-bit filtered-parity construction and a space-sensitive quantum lower bound based on compressed-oracle capacity and a new parity-to-capacity inequality, which may be of independent interest.
\end{abstract}

\newpage
\tableofcontents
\newpage

\section{Introduction}

Query complexity, which asks how much of an input an algorithm must inspect to
solve a problem, plays an important role in complexity theory. 
It provides a clean setting for understanding the power and limits of both
classical and quantum algorithms, and thus provides an extremely important
common framework for identifying quantum speedups. 
Based on the intuition that restrictions on workspace affect the time required
for computation and the fact that quantum memories are valuable, quantum
query--space tradeoffs have become an established topic in complexity theory.
At first glance, it seems trivial to state that bounded-space computation
requires more queries. 
Nevertheless, a substantial body of work establishes query--space tradeoffs for
fundamental classical problems, including long-output problems, that is,
problems with many bits of output, such as sorting \cite{BFK+81,BC82}, finding
unique elements \cite{Bea91}, and multi-collision finding \cite{Din20}, as well
as short-output or even Boolean problems, such as branching programs
\cite{BJS01}, randomized decision problems \cite{BSSV03}, and element
distinctness \cite{BFM+87,Yao94}. 
However, in the quantum setting, almost all results that we have for quantum
query--space tradeoffs seem to apply only to problems with long outputs, such as
sorting~\cite{Kla03,KSdW07}, function inversion~\cite{CGLQ20}, finding multiple
collision pairs~\cite{HM23}, and matrix problems~\cite{KSdW07,BKW24}.
As for short-output problems, Bera and SAPV~\cite{BS24} managed to show a
one-bit quantum time--space tradeoff in a restricted model of generalized
quantum branching programs, where each node's outgoing transition vectors for
query answers $0$ and $1$ differ only by a phase.

A natural approach to showing a query--space tradeoff for a \emph{decision}
problem (or any short-output problem) is to start with a problem whose known
query-optimal quantum algorithms all use substantial workspace.
The \emph{collision} problem is a natural candidate.
For a function with domain size $N$, the Brassard--H{\o}yer--Tapp algorithm
finds a collision using $O(N^{1/3})$ queries and $O(N^{1/3})$
qubits~\cite{BHT97}.
Aaronson and Shi subsequently proved a matching $\Omega(N^{1/3})$ query lower
bound, even with unrestricted workspace~\cite{AS04}, yet no near-optimal
algorithms using less space have been found: a Grover-based algorithm uses
only $O(\log N)$ qubits of workspace but makes $O(\sqrt N)$ queries~\cite{Gro96,
  Aar21}.
Aaronson lists finding query--space tradeoffs (or query--space lower bounds) for
the collision problem as one of the main open problems in~\cite{Aar21}, but
little progress has been made since then.
Element distinctness, a problem closely related to collision finding, is equally
important.
Its quantum walk algorithm uses $O(N^{2/3})$ queries and polynomial
workspace~\cite{Amb07}, and this query bound matches the unrestricted quantum
lower bound~\cite{AS04}.

The subtlety is that known methods rely on a progress measure, a technique that
tracks the information acquired by an algorithm.
They bound its increase at each step of the algorithm, showing that few queries
yield little information overall.
Consequently, the longer an answer to the problem is, the more information is
needed to uncover it, and thus the more queries are needed.
These techniques work well only on problems that require a long output and do
not directly extend to short-output problems, let alone Boolean-output ones.
 
Recent work by Hao, Huang, and Liu makes substantial progress~\cite{HHL26}.
They show that even for problems with short (yet not one-bit) outputs, achieving
the optimal quantum query complexity can require a large workspace in the
quantum random oracle model.
Notably, the problem they study, the nested collision problem, is a variant of
the collision problem.
Based on compressed-oracle techniques, they introduce a two-oracle recording
method that transfers time--space tradeoffs from the better-understood
long-output setting to their short-output problem.
However, their result establishes only a query--space tradeoff for this problem
and does not determine whether bounded-space quantum algorithms lose their
advantage over unrestricted classical algorithms.
Hence, beyond establishing a tradeoff for Boolean-output problems, a stronger
question is whether the tradeoff can eliminate the quantum query advantage in
bounded space.
Thus, in this work, we focus on the following question:
\begin{quote}
  \emph{Can the size of the workspace determine whether a quantum query
    advantage for a Boolean function exists?}
\end{quote}

We answer this question affirmatively.
\begin{theorem}[Informal]
\label{thm:intro-main}
There exists a total Boolean function $f$ with domain size $R$ such that if space $S$ satisfies
\begin{equation*}
  S = O(R^{\frac{1}{401}}),
\end{equation*}
then for sufficiently large $n$,
\begin{equation*}
  Q(f)
  <R(f)
  <Q_S(f),
\end{equation*}
where $R$ is bounded-error randomized query complexity, $Q$ is the bounded-error quantum query complexity, and $Q_S$ its bounded-error quantum query complexity with workspace size at most $S$.
\end{theorem}

Our starting point is a filtered-parity problem, which has one-bit output, and predicting the output bit forces a quantum algorithm to record information with many queries.
We formalize this through a new parity-to-capacity inequality and combine it with the compressed-oracle capacity bound of Hao, Huang, and Liu~\cite{HHL26} to obtain a strong query--space tradeoff despite the one-bit output.
Recall that Grover's algorithm gives a quadratic speedup for OR.
The mechanism above could be turned into a worst-case total Boolean function
$f$ satisfying $Q(f)<R(f)<Q_S(f)$, mainly by appending an independent OR block.

\subsection{Main Results}

We first introduce a family of one-bit oracle problems, called $\ell$-Sum-Filtered Parity.
The problem is defined over two independent random oracles: one selects a structured set of $\ell$-tuples, and the other assigns a Boolean value to each tuple.
For an integer $k$, write $[k] = \{0, 1, \ldots, k - 1\}$.

\begin{definition}[$\ell$-Sum-Filtered Parity]\label{def:task}
Fix a constant integer $\ell \ge 4$.
Let
\begin{equation*}
  N = 2^n,
  \qquad
  M = \lfloor N^{\frac{1}{\ell-1}} \rfloor,
\end{equation*}
and let
\begin{equation*}
  H\colon[M]\to[N],
  \qquad
  G\colon{[M]}^\ell\to\{0,1\}
\end{equation*}
be independent uniformly random functions.
Define
\begin{equation*}
  Y_H := \left\{(x_1,\ldots,x_\ell)\in{[M]}^\ell :
  x_1<\cdots<x_\ell,\quad
  \sum_{i=1}^{\ell} H(x_i)=0 \pmod N
  \right\}.
\end{equation*}
Given oracle access to $(H,G)$, the task is to output
$F(H,G):=\bigoplus_{x\in Y_H}G(x)$,
where the parity of the empty set is zero.
\end{definition}

Our first result gives a general query--space tradeoff for this problem.

\begin{theorem}[General query--space tradeoff]
\label{thm:main}
Fix a constant integer $\ell\ge4$.
The $\ell$-Sum-Filtered Parity problem has the following properties.
\begin{enumerate}
    \item
    There is a deterministic classical algorithm that succeeds with
    probability $1-o(1)$, uses $O(M)$ oracle queries, and uses
    $O(M\log N)$ bits of space.

    \item
    For every constant
    $0<\varepsilon\le1/2$, every quantum algorithm that succeeds with
    probability at least $1/2+\varepsilon$, uses $S$ qubits of space,
    and makes $T$ oracle queries satisfies
    \begin{equation*}
      S^{\ell+1}T = \Omega \left(
        N^{\frac{1}{2(\ell-1)}+\frac{1}{\ell+1}} \right),
    \end{equation*}
    whenever $S=\Omega(\log N)$ and
    $S=O(N^{\frac{1}{(\ell+1)(2\ell+1)}})$.
\end{enumerate}
\end{theorem}

The choice $M=\Theta(N^{\frac{1}{\ell-1}})$ balances the number of $H$-queries and selected $G$-entries.
Consequently, within the stated space range, if
$S=o\!\left(N^{\frac{\ell-3}{2(\ell-1){(\ell+1)}^2}}\right)$,
then every such quantum algorithm makes
$\omega(N^{\frac{1}{\ell-1}})$ queries, whereas the classical algorithm
makes only $O(N^{\frac{1}{\ell-1}})$ queries.

We next turn this separation into one for an explicit total Boolean function.

\begin{definition}[$\ell$-Sum-OR Parity]\label{def:intro-sum-or}
Let $N=2^n$ and $M=\lfloor N^{1/(\ell-1)}\rfloor$.
Let $H\in{\{0,1\}}^{[M]\times[n]}$,
$G\in{\{0,1\}}^{{[M]}^\ell}$, and
$z\in{\{0,1\}}^L$ be arbitrary input blocks, where
$L=\lceil M^{1+(\ell-3)/(4(\ell+1))}\rceil$.
For $x\in[M]$, let $H(x)\in[N]$ be the integer encoded by
$(H_{x,1},\ldots,H_{x,n})$, and similarly for $G(x)$.
Let
\begin{equation*}
  Y_H =
  \left\{
  (x_1,\ldots,x_\ell)\in{[M]}^\ell :
  x_1<\cdots<x_\ell,\quad
  \sum_{i=1}^{\ell} H(x_i)=0 \pmod N
  \right\}.
\end{equation*}
Given query access to these Boolean blocks, the required output is determined as
follows.
If
\begin{equation*}
  |Y_H|
  \le
  \left\lceil
    2\binom{M}{\ell}/N
  \right\rceil,
\end{equation*}
then
\begin{equation*}
  \mathsf{SumOR}_{\ell,n}(H, G,z)
  =
  \left(
    \bigoplus_{\mathbf{x}\in Y_H}G(\mathbf{x})
  \right)
  \oplus\operatorname{OR}(z).
\end{equation*}
Otherwise,
\begin{equation*}
  \mathsf{SumOR}_{\ell,n}(H,G,z)
  =
  \operatorname{OR}(z).
\end{equation*}
Here we set the parity of the empty set to be zero.
\end{definition}

Our main result gives the desired separation between unrestricted quantum,
randomized classical, and space-bounded quantum query complexity.

\begin{theorem}[Total quantum speedup and small-space reversal]
\label{thm:intro-sum-or}
Fix a constant integer $\ell\ge4$.
For all sufficiently large $n$,
\begin{equation*}
  D(\mathsf{SumOR}_{\ell,n})
  =R(\mathsf{SumOR}_{\ell,n})
  =\Theta\!\left(M^{1+\frac{\ell-3}{4(\ell+1)}}\right),
  \qquad
  Q(\mathsf{SumOR}_{\ell,n})=\widetilde\Theta(M).
\end{equation*}
Moreover, every $S$-qubit quantum algorithm with constant advantage and $T$
queries satisfies
\begin{equation*}
  S^{\ell+1}T
  =\widetilde{\Omega}\!\left(M^{3/2-2/(\ell+1)}\right)
\end{equation*}
for $ \Omega(\log M)
  \le S
  \le O\left(M^{\frac{\ell-1}{(\ell+1)(2\ell+1)}}\right)$.
Consequently, for every fixed
$0<\eta<\frac{\ell-3}{4{(\ell+1)}^2}$,
if $ S
  \le O(M^{\frac{\ell-3}{4{(\ell+1)}^2}-\eta})$,
then
\begin{equation*}
  Q(\mathsf{SumOR}_{\ell,n})
  <R(\mathsf{SumOR}_{\ell,n})
  <Q_S(\mathsf{SumOR}_{\ell,n}).
\end{equation*}
\end{theorem}

In particular, setting $\ell=4$ gives an explicit family of total Boolean
functions $f$ satisfying
\begin{equation*}
  Q(f)=\widetilde\Theta(M),
  \qquad
  R(f)=D(f)=\Theta(M^{21/20}),
  \qquad
  S^5Q_S(f)=\widetilde{\Omega}(M^{11/10}).
\end{equation*}
Therefore, for every fixed $0<\eta<1/100$ and
$ S
  \le
 O(M^{1/100-\eta})$,
we have
$Q(f)<R(f)<Q_S(f)$.


\subsection{Technique Overview}

The classical upper bound is straightforward.
The algorithm queries and stores all values of $H$, enumerates the filtered set $Y_H$, and queries $G$ only on tuples in $Y_H$.
With high probability, $|Y_H|=\Theta(M)$, which gives the claimed query and space complexities.

Assume for now that all queries to $G$ are classical.
Let $K$ count the distinct tuples $x\in Y_H$ on which an algorithm $\calA$
queries $G$.
The intuition of Hao, Huang, and Liu~\cite{HHL26} is that finding a collision
with high probability requires locating many tuples in $Y_H$ and explicitly
querying $G$ on them.
Since $Y_H$ is a random set specified by $H$, locating these tuples requires
time and space.
Thus, $K$ connects the collision requirement to a query--space lower bound.
The upper bound on $K$ applies to any problem with access to $(H,G)$.
An output that depends heavily on $G|_{Y_H}$ should, in turn, force large $K$
for any algorithm with high success probability.
This suggests a natural one-bit problem, $\ell$-Sum-Filtered Parity problem,
which takes $G$ to be Boolean ($N_0=2$) and outputs $\bigoplus_{x\in Y_H}G(x)$.

For fixed $H$, exact computation with classical queries to $G$ requires querying
every $x\in Y_H$, so $K\ge|Y_H|$.
Quantum queries may be in superposition, so counting distinct queried tuples is
no longer appropriate.
We instead use Zhandry's compressed oracle~\cite{Zha19} to measure the expected
recording weight on database entries indexed by $Y_H$, again denoted by $K$.
As in~\cite{HHL26}, a progress measure bounds this capacity in terms of time and
space.
Our main new step is to show that predicting the parity with constant advantage
nevertheless forces the capacity $K$ to be large.

Combining our new lower bound with the existing upper bound on the capacity $K$
gives the quantum part of the query--space tradeoff theorem (\cref{thm:main}).

\begin{remark}
  When considering the tradeoff for the total Boolean function
  $\mathsf{SumOR}_{\ell,n}$, the queries are bit-by-bit.
  However, $\ell$-Sum-Filtered Parity is a one-bit oracle problem, and the
  queries, by definition, are tuple-by-tuple.
  The two query models can be translated into each other using a standard
  encoding with logarithmic loss, so this does not fundamentally affect the
  tradeoff.
\end{remark}

\section{Problem Setup}

\subsection{Size of the Filtered Subset}\label{sec:size}

We first analyze the size of the filtered set $Y_H$.
\begin{lemma}\label{lem:size-of-y}
Fix a constant integer $\ell\ge4$.
Then
\begin{equation*}
  \begin{aligned}
    \Prb_H\left[
      \frac{1}{2N}\binom M\ell\le |Y_H|\le M
    \right]
    \ge 1-O(M^{-1}).
  \end{aligned}
\end{equation*}
\end{lemma}

\begin{proof}
Express the size of the set $Y_H$, which is defined in \cref{def:task}, as a
sum of indicator functions.
\begin{equation*}
  \begin{aligned}
    |Y_H| =\sum_{A\in\binom{[M]}\ell}
    \mathds1\left[\sum_{x\in A}H(x)=0\pmod N\right].
  \end{aligned}
\end{equation*}
For each $A\in\binom{[M]}\ell$, condition on $\ell-1$ of the $\ell$ values
$\{H(x):x\in A\}$.
Exactly one value of the remaining uniform variable makes the sum zero
modulo $N$.
Therefore,
\begin{equation*}
  \begin{aligned}
    \E_H \mathds1\left[\sum_{x\in A}H(x)=0\pmod N\right]
    = \Prb_H\left[\sum_{x\in A}H(x)=0\pmod N\right]
    = \frac1N.
  \end{aligned}
\end{equation*}
Thus the expected size is $\E_H|Y_H|=\binom M\ell/N$.

For distinct $A,B\in\binom{[M]}\ell$, choose $a\in A\setminus B$ and
$b\in B\setminus A$.
Condition on every $H(x)$ except $H(a)$ and $H(b)$.
The zero-sum equations for $A$ and $B$ then determine $H(a)$ and $H(b)$,
respectively.
Hence
\begin{equation*}
  \Prb_H\left[
    \sum_{x\in A}H(x)=0\pmod N
    \ \text{and}\
    \sum_{x\in B}H(x)=0\pmod N
  \right]
  =\frac1{N^2}.
\end{equation*}
Thus, the covariance terms for distinct $A$ and $B$ vanish.
It follows that
\begin{equation*}
  \begin{aligned}
    \operatorname{Var}_H(|Y_H|)
    &=\sum_{A\in\binom{[M]}\ell}
      \operatorname{Var}_H\left(
        \mathds1\left[\sum_{x\in A}H(x)=0\pmod N\right]
      \right)\\
    &=\binom M\ell\frac1N\left(1-\frac1N\right)
    \le\binom M\ell\frac 1N.
  \end{aligned}
\end{equation*}
Define $\mu:=\E_H|Y_H|=\binom M\ell \frac 1N$.
The identity $M=\lfloor N^{\frac{1}{\ell-1}}\rfloor$ implies
$M^{\ell-1}\le N<{(M+1)}^{\ell-1}$.
Thus, $N=\Theta(M^{\ell-1})$.
Since $\ell$ is fixed, $\binom M\ell=\Theta(M^\ell)$, so
$\mu=\Theta(M)$.
Moreover, $N\ge M^{\ell-1}$ and
$\binom M\ell\le M^\ell/\ell!$, which give $\mu\le M/\ell!$.
Thus, Chebyshev's inequality and the variance bound above give
\begin{align*}
  \Prb_H[|Y_H|<\mu/2]
  &\le\frac{4\operatorname{Var}_H(|Y_H|)}{\mu^2}\\
  &\le\frac4\mu\\
  &=O(M^{-1}).
\end{align*}
For the upper tail, the bound on $\mu$ gives
$M-\mu\ge(1-1/\ell!)M$.
Applying Chebyshev's inequality again yields
\begin{align*}
  \Prb_H[|Y_H|>M]
  &\le
  \frac{\operatorname{Var}_H(|Y_H|)}
       {{(M-\mu)}^2}\\
  &\le
  {(1-1/\ell!)}^{-2}\frac{\mu}{M^2}\\
  &=O(M^{-1}).
\end{align*}
A union bound now proves the claim.

\end{proof}

\subsection{Compressed Oracle and the Computational Model}

In this section, we explain how queries in the standard quantum random-oracle
model (QROM) can be ``recorded'' using Zhandry's compressed-oracle
technique~\cite{Zha19}.
We focus on the representation needed for
\textsc{$\ell$-Sum-Filtered Parity}.

We use the following four registers.
\begin{itemize}
    \item $\Qreg$ contains the query index and answer/target registers.
    \item $\Wreg$ contains all working registers except for $\Qreg$.
\end{itemize}
Together, $\Wreg$ and $\Qreg$ are working registers used by the algorithm.
They are initialized as $\ket 0$ before the computation, and the joint
register $\Wreg\Qreg$ contains at most $S$ qubits.
The two remaining registers, while not considered as working registers, play
an important role in compressed-oracle analysis.
\begin{itemize}
    \item $\Hdbreg$ is the database register that stores the random oracle
    $H$ in the compressed oracle view.
    \item $\Gdbreg$ is the database register that stores the random oracle
    $G$ in the compressed oracle view.
\end{itemize}
The compressed standard oracle can be viewed as a type of lazy sampling
technique.
Instead of designating a particular $H$ or $G$ at the very beginning, the
compressed oracle maintains databases on $\Hdbreg, \Gdbreg$ in superposition.
Oracle queries are hence interactions between working registers and database
registers.
Informally speaking, $\Hdbreg, \Gdbreg$ would be initialized in uniform
superpositions of the states defined as follows,
\begin{equation*}
  \ket{H}_{\Hdbreg} :=
  \bigotimes_{j\in [M]}\ket{H(j)}_{\Hdbreg_j}, \quad \ket{G}_{\Gdbreg}
  :=\bigotimes_{\mathbf{j}\in{[M]}^{\ell}}
  \ket{G(\mathbf{j})}_{\Gdbreg_{\mathbf{j}}}.
\end{equation*}

\paragraph{Standard oracle.}

Registers $\Hdbreg$ and $\Gdbreg$ are initialized to
$\sum_H\ket{H}_{\Hdbreg}$ and $\sum_G\ket{G}_{\Gdbreg}$, respectively,
ignoring the normalizing factors.
An algorithm accesses the oracle $H$ coherently by applying the unitary
$O_H^{\mathrm{std}}$, which acts nontrivially only on $\Qreg$ and $\Hdbreg$.
Specifically,
\begin{equation*}
  O_H^{\mathrm{std}} \ket{x,z}_\Qreg\ket{H}_\Hdbreg :=
  \ket{x,z\oplus H(x)}_\Qreg\ket{H}_\Hdbreg.
\end{equation*}
The same discussion applies to $G$, except that $G$'s domain and codomain have
different sizes.
Note that an $H$-query requires $O(\log M + \log N) = O(\log N)$ qubits on
$\Qreg$, while a $G$-query requires
$O(\ell\log M + 1) = O(\log N)$ qubits because $\ell$ is fixed.
Hence they both can be included in an $O(\log N)$ query workspace.
We therefore slightly abuse notation and use $\Qreg$ for the query register
in both cases.
Then
\begin{equation*}
  O_G^{\mathrm{std}} \ket{\mathbf{x},b}_{\Qreg}\ket{G}_\Gdbreg :=
  \ket{\mathbf x,b\oplus G(\mathbf x)}_{\Qreg}\ket{G}_\Gdbreg.
\end{equation*}
This is the standard form of an oracle query, and Zhandry~\cite{Zha19} showed
that it can also be viewed as the following.

\paragraph{Compressed oracle.}
Write $\ket{\widehat 0}:=(\ket 0+\ket 1)/\sqrt 2$ and
$\ket{\widehat 1}:=(\ket 0-\ket 1)/\sqrt 2$ for the Fourier-basis states
of the answer qubit.
In this basis, an $H$-query acts non-trivially only on the registers $\Qreg$
and $\Hdbreg$, as follows:
\begin{equation*}
  O_H\ket{ x,\widehat z}_{\Qreg}\ket{H}_\Hdbreg
  ={(-1)}^{z\cdot H( x)} \ket{ x,\widehat z}_{\Qreg}\ket{H}_\Hdbreg.
\end{equation*}
Unlike the standard XOR queries, the Fourier-basis query changes a phase, and
hence also ``acts on'' the database registers.
Moreover, if we encode the database registers $\Hdbreg$ and $\Gdbreg$ in the
basis $\ket{\widehat 0},\ket{\widehat 1}$, they encode information about the
query more clearly.

Let $\mathbf{0} = {00\cdots 0}$ be the all-zero Boolean string.
Initially, ignoring the normalizing factor, the $\Hdbreg$ register is
$\sum_H\ket{H}_\Hdbreg=\ket{\widehat{\mathbf 0}}_\Hdbreg$.
In this representation, an $H$-phase-query acts as follows:
\begin{equation*}
  \begin{aligned}
    &O_H\ket{x, z}_{\Qreg}\otimes
      \ket{\widehat{s}}_{\Hdbreg_x}\ket{\widehat S}_{\Hdbreg_{-x}}\\
    &=
      \ket{x, z}_{\Qreg}\otimes
      \ket{\widehat{s\oplus z}}_{\Hdbreg_x}
      \ket{\widehat S}_{\Hdbreg_{-x}}.
  \end{aligned}
\end{equation*}
Here $\Hdbreg_x$ is the database cell register at position $x\in [M]$, while
$\Hdbreg_{-x}$ contains the rest.

A similar argument applies to $G$.
\begin{equation}
  \begin{aligned}
    &O_G\ket{\mathbf x, b}_{\Qreg}\otimes
      \ket{\widehat{p}}_{\Gdbreg_{\mathbf x}}
      \ket{\widehat P}_{\Gdbreg_{-\mathbf x}}\\
    &=
      \ket{\mathbf x, b}_{\Qreg}\otimes
      \ket{\widehat{p\oplus b}}_{\Gdbreg_{\mathbf x}}
      \ket{\widehat P}_{\Gdbreg_{-\mathbf x}}.
  \end{aligned}
\end{equation}
\begin{theorem}[{\cite{Zha19}} and
  {\cite[Lemma 4.9 in the full version]{HHL26}}]
Let $\calA$ be an (unbounded) quantum algorithm making oracle queries.
The output of $\calA$ given access to the compressed oracles $O_H, O_G$ is identical to
the output of $\calA$ given access to the standard oracles
$O_H^{\mathrm{std}}, O_G^{\mathrm{std}}$.
\end{theorem}

\section{Classical Upper Bound}\label{sec:classical}

In this section, we give a classical upper bound for
\textsc{$\ell$-Sum-Filtered Parity} using the following algorithm.
\begin{enumerate}[label=(\arabic*)]
\item Query and store $H(x)$ for every $x\in[M]$.
\item Enumerate the $\ell$-tuples in $Y_H$, query the corresponding
$G$-bits, and compute their parity, stopping with output $0$ if more
than $M$ such tuples are found.
\end{enumerate}
This requires $O(M\log N)$ space and gives a worst-case
$O(M)$ query bound.
By \cref{lem:size-of-y}, the cutoff changes the output only with
probability $O(M^{-1})$.

\begin{algorithm}[H]
\caption{Classical algorithm for
  \textsc{$\ell$-Sum-Filtered Parity}}\label{alg:classical}
\begin{algorithmic}[1]
\State{} Query and store $H(x)$ for every $x\in[M]$.
\State{} $\mathit{answer}\gets 0$ and $\mathit{count}\gets 0$.
\ForAll{$x_1<\cdots<x_\ell$ in $[M]$}
  \If{$\sum_{i=1}^\ell H(x_i)=0\pmod N$}
    \If{$\mathit{count}=M$}
      \State{} \Return{} $0$
    \EndIf{}
    \State{} Query $G(x_1,\ldots,x_\ell)$ and XOR it into
      $\mathit{answer}$.
    \State{} $\mathit{count}\gets\mathit{count}+1$.
  \EndIf{}
\EndFor{}
\State{} \Return{} $\mathit{answer}$
\end{algorithmic}
\end{algorithm}

\begin{proposition}\label{prop:classical}
Fix a constant integer $\ell\ge4$.
There is a deterministic classical algorithm that outputs $F(H,G)$ with
probability $1-O(M^{-1})$, using $M$ queries to $H$ and at most
$M$ queries to $G$ with $O(M\log N)$-bit space.
\end{proposition}

\begin{proof}
If $|Y_H|\le M$, \cref{alg:classical} queries every $G$-bit
appearing in $F(H,G)$ exactly once and returns their parity.
Thus, it can fail only when $|Y_H|>M$, which has probability
$O(M^{-1})$ by \cref{lem:size-of-y}.
The algorithm makes $M$ queries to $H$ and at most $M$ queries to $G$.
The stored $H$-table uses $M\log N$ bits; the loop indices, counter, query
registers, and output bit use $O(\log N)$ additional space.
The algorithm therefore has space complexity $O(M\log N)$.
\end{proof}

\section{Parity Forces Recording Capacity}
\label{sec:capacity-lower}

Fix a quantum algorithm $\mathcal A$.
We show that any advantage in predicting the parity on the $H$-selected domain,
i.e., solving \textsc{$\ell$-Sum-Filtered Parity}, forces a large expected
recording weight on the database register $\Gdbreg$ for the random oracle $G$.

\subsection{Introducing the capacity.}

We first specify what ``large expected recording weight'' means.
Let $\calA$ be the algorithm that aims to solve $\ell$-Sum-Filtered Parity.
For fixed $H$, we can omit $\Hdbreg$ as it stores unentangled information
about $H$.
During the computation, $\calA$ queries $G$ by applying $O_G$, and the
inter-query operations may depend on $H$.
Let $q\le T$ be the number of $G$-queries, and let
$\Lambda_0^{(H)},\ldots,\Lambda_q^{(H)}$ be the corresponding inter-query
quantum channels on $\Wreg\Qreg$.
Each $\Lambda_i^{(H)}$ is implicitly extended by the identity on
$\Gdbreg$.
Finally, $\calA$ performs a two-outcome POVM on the working registers
$\Wreg\Qreg$.
Thus, immediately before the final measurement, the joint state is
\begin{equation}
  \label{eq:recording-final-state}
  \begin{aligned}
    \rho_{H,0}
    &:= \Lambda_0^{(H)}
    \left(
      \proj{0}_{\Wreg\Qreg}
      \otimes\proj{\widehat 0}_{\Gdbreg}
    \right),\\
    \rho_{H,i+1}
    &:= \Lambda_{i+1}^{(H)}
    \left(
      O_G\rho_{H,i}O_G^\dagger
    \right),
    \qquad 0\le i<q,\\
    \rho_H&:=\rho_{H,q}.
  \end{aligned}
\end{equation}
Here and below, identities on omitted registers are implicit:
$\Lambda_i^{(H)}$ acts on $\Wreg\Qreg$, whereas $O_G$ acts on
$\Qreg\Gdbreg$.

For $Y\subseteq{[M]}^\ell$, define the parity-flip operator
\begin{equation}
  \label{eq:parity-flip-operator}
  X_Y:=\prod_{\mathbf x\in Y}X_{\Gdbreg_{\mathbf x}},
\end{equation}
and the recording-weight observable
\begin{equation}
  \label{eq:recording-observable}
  W_Y:=\sum_{\mathbf x\in Y}\proj{\hat 1}_{\Gdbreg_{\mathbf x}}.
\end{equation}

We next define the recording capacity that connects the algorithm's success
bias to the HHL capacity bound.

\begin{definition}[Selected recording capacity]
\label{def:selected-capacity}
First, fix an algorithm $\calA$ and a choice of $H\colon[M]\to[N]$.
Define
\begin{equation*}
  K_H(\calA)
  :=
  \operatorname{Tr}\!\left[
    \rho_H
    \bigl(I_{\Wreg\Qreg}\otimes W_{Y_H}\bigr)
  \right],
\end{equation*}
and define the selected recording capacity of $\calA$ as
$K(\calA):=\E_H[K_H(\calA)]$,
where the expectation is over a uniformly random $H$.
\end{definition}

\subsection{The success probability as an observable value}

\begin{lemma}
\label{lem:parity-correlation}
Represent the final binary POVM of $\calA$ by a Hermitian operator $B$ on
$\Wreg\Qreg$ with $\lVert B\rVert\le1$. 
Define the conditional success bias as
$\beta_H:=2\Prb_{G,\mathcal A}[\mathcal{A}^{H,G}=F(H,G)]-1$.
Then
\begin{equation*}
  \beta_H =
  \operatorname{Tr}\!\left[
    \rho_H\bigl(B\otimes X_{Y_H}\bigr)
  \right].
\end{equation*}
\end{lemma}

\begin{proof}
For each $G$, let $\rho_{H,G}$ be the
premeasurement density operator on $\Wreg\Qreg$ when the Boolean oracle is
$G$.
Before Fourier transforming the truth-table register $\Gdbreg$, let
$\sigma_H$ be the joint density operator.
Its diagonal blocks satisfy
\begin{equation*}
  \bra{G}_{\Gdbreg}\sigma_H\ket{G}_{\Gdbreg}
  =2^{-M^\ell}\rho_{H,G}.
\end{equation*}
For fixed $G$, the definition of $B$ gives
\begin{equation*}
  2\Prb_{\mathcal A} [\mathcal A^{H,G}=F(H,G)]-1
  = {(-1)}^{F(H,G)} \operatorname{Tr}(B\rho_{H,G}).
\end{equation*}
Because $G$ is uniform, averaging this identity over $G$ gives
\begin{equation*}
  \begin{aligned}
    \beta_H
    &=2^{-M^\ell} \sum_{G\colon{[M]}^\ell\to\{0,1\}}
      \left(2\Prb_{\mathcal A}[\mathcal{A}^{H,G}=F(H,G)]-1
      \right)\\
    &=2^{-M^\ell} \sum_{G\colon{[M]}^\ell\to\{0,1\}}
      {(-1)}^{F(H,G)}\operatorname{Tr}(B\rho_{H,G}).
  \end{aligned}
\end{equation*}
For each $\mathbf x$, let $Z_{\Gdbreg_{\mathbf x}}$ act on the
truth-table qubit indexed by $\mathbf x$.
Since $F(H,G)=\bigoplus_{\mathbf x\in Y_H}G(\mathbf x)$,
\begin{equation*}
  \left(\prod_{\mathbf x\in Y_H}Z_{\Gdbreg_{\mathbf x}}
  \right)
  \ket{G}_{\Gdbreg} ={(-1)}^{F(H,G)}\ket{G}_{\Gdbreg}.
\end{equation*}
Substituting the diagonal-block identity above and using the orthogonality
of the states $\ket{G}_{\Gdbreg}$ now yields
\begin{equation*}
  \begin{aligned}
    &\operatorname{Tr}\!\left[
      \sigma_H \left(
        B\otimes
        \prod_{\mathbf x\in Y_H}Z_{\Gdbreg_{\mathbf x}}
      \right)
    \right]\\
    &\quad=2^{-M^\ell}
      \sum_{G\colon{[M]}^\ell\to\{0,1\}}
      {(-1)}^{F(H,G)}\operatorname{Tr}(B\rho_{H,G})\\
    &\quad=\beta_H.
  \end{aligned}
\end{equation*}
Finally, the Hadamard transform changes $\sigma_H$ into $\rho_H$ and maps
each $Z_{\Gdbreg_{\mathbf x}}$ to
$X_{\Gdbreg_{\mathbf x}}$.
Hence, by \cref{eq:parity-flip-operator},
\begin{equation*}
  \beta_H
  =\operatorname{Tr}\!\left[
    \rho_H \left(
      B\otimes
      \prod_{\mathbf x\in Y_H}X_{\Gdbreg_{\mathbf x}}
    \right)
  \right]
  =\operatorname{Tr}\!\left[
    \rho_H(B\otimes X_{Y_H})
  \right].
\end{equation*}
\end{proof}

\subsection{An exact parity-to-capacity inequality}
\label{sec:parity-capacity-inequality}


\begin{lemma}
\label{lem:parity-to-capacity}
For every fixed $H$, $K_H(\calA)\ge |Y_H|\beta_H^2/4$.
\end{lemma}

\begin{proof}
Since we are working with fixed $H$, write
$\rho:=\rho_H$, $Y:=Y_H$, $K:=K_H(\calA)$, and
$\beta:=\beta_H$.
If $Y=\varnothing$, then $K=0$ and the claim is immediate.
Assume that $Y\ne\varnothing$, and define the Hermitian operators
\begin{equation*}
  O_1 :=I_{\Wreg\Qreg}\otimes
    \left(I_{\Gdbreg}-\frac{2}{|Y|}W_Y\right),
  \qquad
  O_2:=B\otimes X_Y.
\end{equation*}
The spectrum of $W_Y$ is contained in $\{0,\ldots,|Y|\}$.
Together with $\lVert B\rVert\le1$ and $X_Y^2=I_{\Gdbreg}$, this gives
\begin{equation*}
  O_1^2\preceq I,
  \qquad
  O_2^2=B^2\otimes I_{\Gdbreg}\preceq I.
\end{equation*}

The operator $X_Y$ flips each recording qubit indexed by $Y$, so
$X_Y W_Y X_Y=|Y|I_{\Gdbreg}-W_Y$.
Consequently,
\begin{equation*}
  \begin{aligned}
    &X_Y \left(I_{\Gdbreg}-\frac{2}{|Y|}W_Y\right)
      X_Y\\
    &\quad=
      I_{\Gdbreg} -\frac{2}{|Y|}\left(|Y|I_{\Gdbreg}-W_Y\right)\\
    &\quad=-\left(I_{\Gdbreg}-\frac{2}{|Y|}W_Y\right).
  \end{aligned}
\end{equation*}
This identity gives $O_1O_2+O_2O_1=0$.

For $a,b\in\mathbb R$, set $O_{a,b}:=aO_1+bO_2$.
The preceding operator inequalities and the anticommutation relation imply
\begin{equation*}
  \begin{aligned}
    O_{a,b}^2 &=a^2O_1^2+b^2O_2^2
      +ab(O_1O_2+O_2O_1)\\
    &\preceq(a^2+b^2)I.
  \end{aligned}
\end{equation*}
By the definitions of $K$ and $\beta$,
\begin{equation*}
  \operatorname{Tr}(\rho O_1)=1-\frac{2K}{|Y|},
  \qquad
  \operatorname{Tr}(\rho O_2)=\beta.
\end{equation*}
Since $O_{a,b}$ is Hermitian, Cauchy--Schwarz gives
\begin{equation*}
  \begin{aligned}
    \left|a\left(1-\frac{2K}{|Y|}\right)+b\beta\right|^2
    &=\left|\operatorname{Tr}(\rho O_{a,b})\right|^2\\
    &\le\operatorname{Tr}(\rho O_{a,b}^2)\\
    &\le a^2+b^2.
  \end{aligned}
\end{equation*}
By Cauchy--Schwarz again,
\begin{equation*}
  \max_{a^2+b^2=1}
  \left|a\left(1-\frac{2K}{|Y|}\right)+b\beta\right|^2
  ={\left(1-\frac{2K}{|Y|}\right)}^2+\beta^2.
\end{equation*}
Combining this identity with the preceding bound yields
$(1-2K/|Y|)^2+\beta^2\le1$.
In particular, $K\ge |Y|(1-\sqrt{1-\beta^2})/2$.
Finally,
\begin{equation*}
  1-\sqrt{1-\beta^2}
  =\frac{\beta^2}{1+\sqrt{1-\beta^2}}
  \ge\frac{\beta^2}{2},
\end{equation*}
which proves the desired inequality.
\end{proof}


\subsection{Averaging over \texorpdfstring{$H$}{H}}
\label{sec:capacity-averaging}

Using \cref{def:selected-capacity}, we now combine the
fixed-$H$ inequality with the concentration of $Y_H$ to lower-bound
$K(\mathcal A)$.

\begin{theorem}
\label{thm:success-forces-capacity}
Fix $0<\varepsilon\le1/2$, and let $\mathcal A$ be a quantum algorithm for
\textnormal{\textsc{$\ell$-Sum-Filtered Parity}}.
Let $K(\mathcal A)$ be its selected recording capacity as defined in
\cref{def:selected-capacity}.
If $\mathcal A$ outputs $F(H,G)$ with probability at least
$1/2+\varepsilon$, then, for all sufficiently large $N$,
\begin{equation*}
  K(\mathcal A)
  \ge \binom{M}{\ell}\frac{\varepsilon^2}{8N}
  =\Omega(\varepsilon^2M).
\end{equation*}
\end{theorem}

\begin{proof}
For each $H$, let $\beta_H$ be the conditional success bias from
\cref{lem:parity-correlation}.
Set
\begin{equation*}
  \mu:=\E_H[|Y_H|]=\binom{M}{\ell}\frac{1}{N},
  \qquad
  \mathcal E:=\{H:|Y_H|\ge\mu/2\}.
\end{equation*}
By \cref{lem:size-of-y}, $\Prb_H[\neg\mathcal E]=O(M^{-1})$.
Since $\varepsilon>0$ is fixed, this probability is at most $\varepsilon$
for all sufficiently large $N$.

The success assumption gives
\begin{equation*}
  \E_H[\beta_H] =2\Prb_{H,G,\mathcal A}
    \bigl[\mathcal{A}^{H,G}=F(H,G)\bigr]-1
  \ge2\varepsilon.
\end{equation*}
Since $\beta_H\le1$ for every $H$,
\begin{equation*}
  \begin{aligned}
    \E_H[\mathds{1}_{\mathcal E}\beta_H]
    &=\E_H[\beta_H]
      -\E_H[\mathds{1}_{\neg\mathcal E}\beta_H]\\
    &\ge2\varepsilon-\Prb_H[\neg\mathcal E]\\
    &\ge\varepsilon.
  \end{aligned}
\end{equation*}
Cauchy--Schwarz therefore gives
\begin{equation*}
  \begin{aligned}
    \varepsilon^2
    &\le{\left(\E_H[\mathds{1}_{\mathcal E}\beta_H]\right)}^2\\
    &\le\Prb_H[\mathcal E]\,
      \E_H[\mathds{1}_{\mathcal E}\beta_H^2]\\
    &\le\E_H[\mathds{1}_{\mathcal E}\beta_H^2].
  \end{aligned}
\end{equation*}

\cref{lem:parity-to-capacity} gives
$K_H\ge |Y_H|\beta_H^2/4$.
Using \cref{def:selected-capacity} and restricting the
expectation to $\mathcal E$, we obtain
\begin{equation*}
  \begin{aligned}
    K(\mathcal A)
    &=\E_H[K_H]\\
    &\ge\frac14\E_H[|Y_H|\beta_H^2]\\
    &\ge\frac{\mu}{8}\E_H[\mathds{1}_{\mathcal E}\beta_H^2]\\
    &\ge\frac{\mu\varepsilon^2}{8}.
  \end{aligned}
\end{equation*}
Substituting $\mu=\binom{M}{\ell}/N$ proves the explicit bound.
Since $M=\lfloor N^{\frac{1}{\ell-1}}\rfloor$, we also have
$\mu=\Theta(M)$, which proves the asymptotic bound.
\end{proof}

\section{Capacity Upper Bound from HHL}

We show that a bounded-space, bounded-query algorithm cannot produce large
capacity $K(\calA)$.

\begin{theorem}[Selected-capacity upper bound]
\label{thm:selected-capacity-upper-bound}
Fix a constant integer $\ell\ge 4$, and let $\calA$ be a uniform quantum
algorithm using at most $S$ qubits and $T$ oracle queries.
If $S=\Omega(\log N)$ and
  $S=O\!\left(
    N^{1/((\ell+1)(2\ell+1))}
  \right)$,
then
\begin{equation*}
  K(\calA)
  =O\!\left(
    S^{2\ell+2}T^2N^{-2/(\ell+1)}
  \right).
\end{equation*}
In particular, for $\ell=4$,
$K(\calA)=O\!\left(S^{10}T^2N^{-2/5}\right)$ whenever
$S=\Omega(\log N)$ and $S=O(N^{1/45})$.
\end{theorem}

\begin{proof}
We use the following one-bit specialization of HHL's capacity bound. 
For each $H$, define the unrestricted ordered zero-sum set
\begin{equation*}
  \overline Y_H :=\left\{
    \mathbf x=(x_1,\ldots,x_\ell)\in{[M]}^\ell:
    \sum_{i=1}^{\ell}H(x_i)=0\pmod N
  \right\}.
\end{equation*}
This is the set used by the HHL capacity for the constant target sequence
$\{y_H\}$ given by $y_H=0$ for every $H$.
Unlike $Y_H$, the set $\overline Y_H$ allows repeated coordinates and all
orderings.
The strict-ordering condition in \cref{def:task} gives
$Y_H\subseteq\overline Y_H$.
Let $\overline K$ be the recorded capacity obtained from $K$ in
\cref{def:selected-capacity} by replacing $Y_H$ with $\overline Y_H$.

\begin{theorem}[One-bit specialization of {\cite[Theorem 9.3 in the full
    version]{HHL26}}]
\label{thm:K-time-space-upperbound}
Fix a constant integer $\ell\ge4$, and let $H$ and $G$ be the two random
oracles in \textnormal{\textsc{$\ell$-Sum-Filtered Parity}}.
Let $\calA$ be a uniform quantum algorithm using $S$ qubits, at most
$T$ queries to $H$, and at most $T$ queries to $G$.
Then, we have\footnote{In~\cite{HHL26}'s original notation, $N_0=2$ and
  hence $n_0=\log N_0=1$, so their bound reduces to the special form.}
\begin{equation*}
  \overline K(\calA)
  =O\!\left(
    S^{2\ell+2}T^2N^{-2/(\ell+1)}
  \right)
\end{equation*}
for $S=\Omega(\log N)$ and
$S=O(N^{1/((\ell+1)(2\ell+1))})$.
\end{theorem}

For every $H$,
\begin{equation*}
  W_{\overline Y_H}-W_{Y_H}
  =\sum_{\mathbf x\in\overline Y_H\setminus Y_H}
    \proj{1}_{\Gdbreg_{\mathbf x}}
  \succeq0.
\end{equation*}
Taking the trace of this operator against $\rho_H$ and then averaging over
$H$ gives $K(\calA)\le\overline K(\calA)$.

Thus, \cref{thm:K-time-space-upperbound} gives
\begin{equation*}
  K(\calA)
  =O\!\left(
    S^{2\ell+2}T^2N^{-2/(\ell+1)}
  \right).
\end{equation*}
Substituting $\ell=4$ gives
$2\ell+2=10$, $2/(\ell+1)=2/5$, and
$(\ell+1)(2\ell+1)=45$, which proves the specialization.
\end{proof}

\section{Proof of the Separation}

\begin{proof}[Proof of \cref{thm:main}]
For part~(i), check \cref{prop:classical}.

For part~(ii), fix $0<\varepsilon\le 1/2$ and a quantum algorithm
$\mathcal A$ as in the theorem.
By \cref{lem:size-of-y} and \cref{thm:success-forces-capacity},
\begin{equation*}
  K(\mathcal A)
  \ge \binom{M}{\ell} \frac{\varepsilon^2}{8N}
  =\Omega\!\left(N^{\frac{1}{\ell-1}}\right).
\end{equation*}

By \cref{thm:selected-capacity-upper-bound},
\begin{equation*}
  K(\mathcal A)
  =O\!\left(
    S^{2\ell+2}T^2N^{-2/(\ell+1)}
  \right),
\end{equation*}
when $S = \Omega(\log N)$ and $S = O(N^{1/((\ell + 1)(2\ell + 1))})$.
Comparing these two bounds yields
\begin{equation*}
  S^{2\ell+2}T^2N^{-2/(\ell+1)}
  =\Omega\!\left(
    N^{\frac{1}{\ell-1}}
  \right).
\end{equation*}
Since $S$ and $T$ are nonnegative, taking square roots gives
\begin{equation*}
  S^{\ell+1}T
  =\Omega\!\left(
    N^{\frac{1}{2(\ell-1)}+\frac{1}{\ell+1}}
  \right). \qedhere
\end{equation*}
\end{proof}

\section{Space-sensitive Quantum Query Advantage for a Total Boolean Function}
\label{sec:sum-or}

\renewcommand{\theHdefinition}{sumor.\arabic{definition}}
\renewcommand{\theHtheorem}{sumor.\arabic{theorem}}

\subsection{Construction and Main Result}

The construction in $\ell$-Sum-OR Parity differs from \FilteredSum{} defined in \cref{def:task}
in three aspects.
\begin{enumerate}
  \item
  We turn the original problem concerning random functions into a worst-case statement about a total function.
  Hence the three input ``oracles'' are now arbitrary truth tables instead of random functions.
  We remark that here the truth table for $H$ should be bit encoded, so querying one whole value $H(x)\in[N]$ costs $\log N$ bit queries.

  \item
  The function discards the filtered-parity term when the filtered set,
  $Y_H$, is too large
  (specifically, when $|Y_H|>\left\lceil 2\binom{M}{\ell}/N\right\rceil$).
  \item
  A disjoint OR block is XORed with the parity block.
\end{enumerate}



\begin{theorem}[Total quantum speedup and small-space reversal]
\label{thm:sum-or-main}
For every sufficiently large integer $n$, the following bounds hold for
\emph{$\ell$-Sum-OR Parity Problem}:
\begin{align*}
  D=R
  &=
  \widetilde\Theta\!\left(
    M^{1+(\ell-3)/(4(\ell+1))}
  \right),\\
  Q
  &=\widetilde{\Theta}(M).
\end{align*}
A deterministic classical algorithm uses $O(M\log M)$ bits of space.
Fix a constant $0<\varepsilon\le1/2$.
Suppose a uniform quantum algorithm uses $S$ qubits, makes $T$ total queries,
and succeeds with probability at least $1/2+\varepsilon$ in the worst case.
Then $T=\omega\!\left(M^{1+(\ell-3)/(4(\ell+1))}\right)$ when
$ \Omega(\log M)  \le S
  \le
  M^{(\ell-3)/(4{(\ell+1)}^2)-\eta}$,
where $\eta$ is a fixed positive constant satisfying
$0<\eta<\frac{\ell-3}{4{(\ell+1)}^2}$.
In fact, the following tradeoff holds:
\begin{equation*}
  S^{\ell+1}T = \widetilde\Omega\!\left(
    \varepsilon M^{3/2-2/(\ell+1)}
  \right).
\end{equation*}
when
$
  S=\Omega(\log M)$
 and
 $ S=O\left(
    M^{(\ell-1)/((\ell+1)(2\ell+1))}
  \right)$.
\end{theorem}

Substituting $\ell=4$ immediately gives the following corollary.

\begin{corollary}
\label{cor:sum-or-four}
For $\ell=4$, let $M=\lfloor N^{1/3}\rfloor$.
Set $L = \left\lceil M^{21/20}\right\rceil$.
For algorithms with constant advantage, if they are space-unrestricted, then
\begin{align*}
  D=R
  &= \Theta(M^{21/20})
  = \Theta(N^{7/20}),
  \\
  Q
  &= \widetilde\Theta(M)
  = \widetilde\Theta(N^{1/3}).
\end{align*}
Hence quantum algorithms are strictly faster than classical ones.
For algorithms using space $S$ and query bound $T$, the following
tradeoff holds:
\begin{equation*}
  S^5T = \widetilde\Omega(M^{11/10})
  = \widetilde\Omega(N^{11/30})
\end{equation*}
when $\Omega(\log M)\le S\le O(M^{1/15})$.
Consequently,
$T=\omega(M^{21/20})=\omega(N^{7/20})$ when
$\Omega(\log M)\le S\le M^{1/100-\eta}$ and $\eta$ is a constant
satisfying $0<\eta<1/100$.
In this space range, the quantum query exceeds the classical one.
\end{corollary}

The rest of this section proves \cref{thm:sum-or-main,cor:sum-or-four}:
\cref{prop:sum-or-unrestricted} gives the unrestricted query complexities,
and \cref{prop:sum-or-tradeoff} gives the bounded-space tradeoff and reversal.

\subsection{Standard Query Facts}

We will use the following standard facts about OR and parity.

\begin{lemma}
\label{lem:sum-or-or-query}
For OR on $L$ bits, $R(\operatorname{OR}_L)=\Omega(L)$ and
$Q(\operatorname{OR}_L)=O(\sqrt L)$.
\end{lemma}

\begin{lemma}
\label{lem:sum-or-parity-query}
For parity on $k$ bits, $Q(\operatorname{PARITY}_k)=\Omega(k)$.
\end{lemma}

We also need to encode the random oracles as Boolean tables, which is specifically listed in the following lemma.

\begin{lemma}\label{lem:bit-to-word-simulation}
Let $F:[K]\to[R]$ be an oracle.
Assume that $K$ and $R$ are powers of two, and set $r=\log R$.
Fix an $r$-bit encoding of the range of $F$.
The standard oracle query maps
\begin{equation*}
  \ket{x,y}
  \mapsto
  \ket{x,y\oplus F(x)}.
\end{equation*}
The corresponding bit-query maps
\begin{equation*}
  \ket{x,j,b}
  \mapsto
  \ket{x,j,b\oplus F_j(x)},
\end{equation*}
where $j\in[r]$ and $F_j(x)$ is the $j$th bit of the encoding of $F(x)$.
Then the following simulations hold.
\begin{enumerate}
  \item
  If a quantum algorithm $\calA$ uses $S$ qubits and $T$ bit queries, then a standard-oracle algorithm $\calB$ has the same output distribution, uses $S+O(r)$ qubits, and $2T$ standard queries.
  \item
  If a quantum algorithm $\calA$ uses $S$ qubits and $T$
  standard queries to $F$, then a bit-query algorithm $\calB$ has the same output distribution, uses $S+O(\log r)$ qubits, and $rT$ bit queries.
\end{enumerate}
\end{lemma}

\begin{proof}
\emph{Bit queries from standard queries.}
The simulator $\calB$ runs the circuit of $\calA$ and simulates each bit query.
To simulate a query to $F_j(x)$, it uses an $r$-qubit scratch register initialized to zero.
It first queries the standard oracle to write $F(x)$ into the scratch register.
It CNOTs the $j$th scratch bit into the answer qubit of the bit query.
It then queries the standard oracle again to restore the scratch register to zero.
On each computational-basis query state, this is exactly
\begin{equation*}
  \ket{x,j,b}
  \mapsto
  \ket{x,j,b\oplus F_j(x)}.
\end{equation*}
Thus every bit query costs at most two standard queries.

\emph{Standard queries from bit queries.}
The simulator $\calB$ again runs the circuit of $\calA$.
To simulate one standard query, it applies the bit-query oracle once for each
$j\in[r]$; the product of these queries is exactly
\begin{equation*}
  \ket{x,y}
  \mapsto
  \ket{x,y\oplus F(x)}.
\end{equation*}
Thus every standard query costs at most $r$ bit queries.
\end{proof}

We've introduced \cref{lem:size-of-y} to analyze the size of the filtered set $Y_H$. 
However, the following corollary would still simplify the analysis of the Boolean version of the problem.

\begin{corollary}
  \label{cor:sum-or-concentration}
Set $\mu:=\binom{M}{\ell}/N$.
For a uniform function $H:[M]\to[N]$,
\begin{equation*}
  \Prb_H\!\left[
    \frac{\mu}{2}
    \le |Y_H|
    \le2\mu
  \right]
  \ge1-\frac{5}{\mu}.
\end{equation*}
\end{corollary}

\begin{proof}
The moment calculation in the proof of \cref{lem:size-of-y} applies here.
It gives $\E_H[|Y_H|]=\mu$ and
$\operatorname{Var}_H(|Y_H|)\le\mu$.
Chebyshev's inequality and a union bound prove the claim.
\end{proof}

\subsection{Unrestricted Query Complexities}

We next determine the three unrestricted query complexities.

\begin{proposition}
\label{prop:sum-or-unrestricted}
The unrestricted query complexities satisfy
\begin{align*}
  D=R
  &=
  \Theta\!\left(
    M^{1+(\ell-3)/(4(\ell+1))}
  \right),
  \\
  Q
  &=\widetilde\Theta(M).
\end{align*}
Moreover, the upper bound for $D$ is achieved by a deterministic algorithm
using $O(M\log M)$ bits.
\end{proposition}

\begin{proof}

\emph{The classical part.}
There is a simple algorithm that shows the upper bound.
It first queries and stores all $M\log N$ bits of $H$, then enumerates $Y_H$
locally and stops counting after $\lceil2\mu\rceil+1$ tuples.
If the count exceeds $\lceil2\mu\rceil$, query every bit of $z$ and return its
OR\@.
Otherwise, enumerate $Y_H$ again and query all selected $G$-bits, XOR their
values, and then XOR the result with the OR of $z$.
This algorithm makes at most $M\log N+\lceil2\mu\rceil+L=\Theta(L)$ queries and
uses $O(M\log M)$ bits of space.
Recall that $L=\lceil M^{1+(\ell-3)/(4(\ell+1))}\rceil$ shows the desired
upper bound.
As for the lower bound, fixing every $G$-bit to zero restricts the function to
OR on $L$ bits.
Applying \cref{lem:sum-or-or-query} to this restriction proves the randomized
query complexity $\Omega(L)$.

\emph{The quantum part.}
For the quantum upper bound, the same algorithm described above still works,
and we can speed the OR part by Grover search.
Hence the quantum query complexity is at most
\begin{equation*}
  O(M\log N+\lceil2\mu\rceil+\sqrt L)
  =
  \widetilde O(M).
\end{equation*}
Here we use the fact that $L=o(M^2)$ for every fixed $\ell\ge4$.

By \cref{cor:sum-or-concentration}, some function $H^*$ satisfies
$\mu/2\le|Y_{H^*}|\le2\mu$.
Hence $Y_{H^*} = \Theta(M)$.
Fix this $H^*$ and fix $z$ to the all-zero string.
The expected output of the resulting function would be parity on $\Theta(M)$
distinct bits.
By \cref{lem:sum-or-parity-query}, it requires $\Omega(M)$ queries.
This proves the quantum lower bound.
\end{proof}

\subsection{The Bounded-Space Reversal}

\begin{lemma}
\label{lem:sum-or-forces-capacity}
Fix $0<\varepsilon\le1/2$.
Let $\calA$ compute $\ell$-Sum-OR Parity with worst-case success probability
at least $1/2+\varepsilon$.
After fixing $z$ to zero, its selected recording capacity satisfies
$K(\calA)=\Omega(\varepsilon^2M)$.
\end{lemma}

\begin{proof}
Consider an $H^*$ satisfying $\mu/2\le|Y_{H^*}|\le2\mu$.
The threshold rule does not fire, so the target is parity on $Y_{H^*}$.
By \cref{lem:parity-correlation} and the inequality in
\cref{lem:parity-to-capacity},
$K_{H^*}(\calA)\ge\mu\varepsilon^2/2$.
Averaging this bound and applying \cref{cor:sum-or-concentration} gives
\begin{equation*}
  K(\calA)
  \ge
  \left(1-\frac{5}{\mu}\right)
  \frac{\mu\varepsilon^2}{2}
  =\Omega(\varepsilon^2M).\qedhere
\end{equation*}
\end{proof}

\begin{proposition}[Bounded-space tradeoff and reversal]
\label{prop:sum-or-tradeoff}
Fix a constant $0<\varepsilon\le1/2$.
Let a uniform quantum algorithm compute $\ell$-Sum-OR Parity with worst-case
success at least $1/2+\varepsilon$.
Suppose that it uses $S$ qubits and makes $T$ total queries.
If
$  S=\Omega(\log M)$
and $
  S=O\left(
    M^{(\ell-1)/((\ell+1)(2\ell+1))}
  \right)$,
then
\begin{equation}
  S^{\ell+1}T = \widetilde\Omega\!\left(
    \varepsilon
    M^{3/2-2/(\ell+1)}
  \right).
  \label{eq:sum-or-tradeoff-derived}
\end{equation}
\end{proposition}

\begin{proof}
Fix $z$ to zero, and omit the queries to the $z$ truth table.
Writing the remaining $G$-queries in the equivalent Hadamard representation
does not change the live space or query count.
Together with the translation by \cref{lem:bit-to-word-simulation},
\cref{thm:selected-capacity-upper-bound} applies to the resulting algorithm.
Combining it with \cref{lem:sum-or-forces-capacity} yields, for an absolute
constant $C$,
\begin{equation*}
  \varepsilon^2M
  \le
  \widetilde O\!\left(
    {(S+C\log N)}^{2\ell+2}T^2N^{-2/(\ell+1)}
  \right).
\end{equation*}
Using $N=\Theta(M^{\ell-1})$ and taking square roots gives
\begin{equation*}
  {(S+C\log N)}^{\ell+1}T
  =
  \widetilde\Omega\!\left(
    \varepsilon
    M^{3/2-2/(\ell+1)}
  \right).
\end{equation*}
Since $S=\Omega(\log M)$ and $\log N=\Theta(\log M)$, this simplifies to
the displayed bound.
This proves the tradeoff.
\end{proof}

\section{Choosing the Parameters}\label{app:sum-or-parameters}

\subsection{Choosing the OR Block}

In \cref{sec:sum-or} we choose the OR block length $L$ to be
$M^{1+(\ell-3)/(4(\ell+1))}$ (without the ceiling).
It is the midpoint of the interval that allows both the unrestricted advantage
and the small-space reversal.
The next proposition shows that in detail.

\begin{proposition}
\label{prop:sum-or-balance}
Replace the relation in the definition by $N=M^{p+o(1)}$, where $0<p<\ell$.
Consider uniform quantum algorithms with advantage
$\varepsilon\in (0, 1/2]$ and space $S$.
Assume that $S=M^{o(1)}$ and $S=\Omega(\log M)$.
Let $c_\ell(p)$ denote the exponent of the deterministic classical query upper bound for
the filtered-parity block.
Let $q_\ell(p)$ denote the exponent yielded by the bounded-space quantum query lower bound.
That is,
\begin{equation*}
  c_\ell(p):=\max\{1,\ell-p\},
  \qquad
  q_\ell(p)
  :=
  \frac{\ell}{2}
  -
  \frac{p(\ell-1)}{2(\ell+1)}.
\end{equation*}
Let the OR block have length $L = M^{\alpha+o(1)}$.
The unrestricted advantage and the small-space reversal both hold when
$c_\ell(p)<\alpha<q_\ell(p)$.
The width of this interval is uniquely maximized at $p=\ell-1$.
At this value, its midpoint is
$\alpha=1+(\ell-3)/(4(\ell+1))$.
\end{proposition}

\begin{proof}
We first identify the two exponents in the statement.
The classical upper bound could be derived from the algorithm we described
before.
\begin{align*}
  M^{c_\ell(p)+o(1)}
  &= M+\binom{M}{\ell}/N\\
  &= M^{\max\{1,\ell-p\}+o(1)}.
\end{align*}
We used the fact that $\binom{M}{\ell}/N=M^{\ell-p+o(1)}$.
This explains the definition of $c_\ell(p)$.
The exponent $q_\ell(p)$ comes from the capacity lower bound.
Specifically,
$M^{q_\ell(p)}=N^{1/(\ell+1)}\sqrt{\binom{M}{\ell}/N}$.
Substituting $N=M^{p+o(1)}$ gives the following exponent.
\begin{align*}
  \frac{\ell-p}{2}
  +\frac{p}{\ell+1}
  &= \frac{\ell}{2} - \frac{p}{2} + \frac{p}{\ell+1}
  \\
  &= \frac{\ell}{2} - \frac{p(\ell-1)}{2(\ell+1)}
  \\
  &= q_\ell(p).
\end{align*}

When $\alpha>c_\ell(p)$, the deterministic algorithm has query cost
$M^{\alpha+o(1)}$.
The OR restriction gives the matching randomized lower bound.
Thus $D=R=M^{\alpha+o(1)}$ in this regime.
The unrestricted quantum algorithm gives the upper bound
\begin{equation*}
  Q
  \le
  M^{\max\{\alpha/2,c_\ell(p)\}+o(1)}.
\end{equation*}
Hence $\alpha>c_\ell(p)$ gives the unrestricted quantum advantage.
For $S=M^{o(1)}$, the capacity lower-bound argument gives
\begin{equation*}
  Q_S
  \ge
  M^{q_\ell(p)-o(1)}.
\end{equation*}
Thus $\alpha<q_\ell(p)$ gives the small-space reversal.
It is therefore enough to choose $c_\ell(p)<\alpha<q_\ell(p)$.
The choice $\alpha = 1 + (\ell-3)/(4(\ell+1))$ and $p = \ell-1$
satisfies these inequalities.
We choose them since the choice of $p$ maximizes the width of the interval
$\big(c_\ell(p), q_\ell(p)\big)$.
Specifically, for $0<p\le\ell-1$,
\begin{equation*}
  q_\ell(p)-c_\ell(p)
  =
  -\frac{\ell}{2}
  +\frac{\ell+3}{2(\ell+1)}p.
\end{equation*}
For $\ell-1\le p<\ell$,
\begin{equation*}
  q_\ell(p)-c_\ell(p)
  =
  \frac{\ell}{2}-1
  -\frac{\ell-1}{2(\ell+1)}p.
\end{equation*}
The first branch is strictly increasing in $p$.
The second branch is strictly decreasing in $p$.
The unique maximum is therefore attained at $p=\ell-1$.
At $p=\ell-1$, we have $c_\ell(\ell-1)=1$ and
$q_\ell(\ell-1)=3/2-2/(\ell+1)$.
Hence we need to choose $\alpha$ such that
$1<\alpha<3/2-2/(\ell+1)$.
So we just choose the midpoint of this interval.
\end{proof}

\subsection{Choosing the Filter Scale}

The construction sets $M=N^{1/(\ell-1)}$ for another reason: namely, to make the
selected set $Y_H$ have the same order of magnitude as the list of $H$-values.
To see this, ignore floors and write $N=M^{p+o(1)}$; then
$\mu:=\E_H[|Y_H|]=\binom{M}{\ell}/N=M^{\ell-p+o(1)}$.
The classical deterministic algorithm that stores $H$ and reads the selected
$G$-bits has parity-block query scale $M+\binom{M}{\ell}/N$,
which is balanced when $\ell-p=1$.
Equivalently, $M=N^{1/(\ell-1)}$.

\section*{Acknowledgements and AI-disclosure}
The authors used large language models as AI-assisted research and writing
tools throughout the preparation of this manuscript.
These tools were used to help brainstorm ideas and explore proof strategies.
Portions of the manuscript text were redrafted or modified with AI assistance
across all sections.
All final mathematical claims, algorithms, proofs, citations, and wording were
reviewed, edited, and validated by the authors.
The authors assume responsibility for all content of the paper.

\bibliographystyle{alphaurl}
\bibliography{main}

\end{document}